\documentclass[11pt]{article}
\usepackage{amsmath,amsthm,amsfonts,amssymb,indentfirst,graphicx,a4,bm,bbm}
\usepackage[latin1]{inputenc}
\usepackage[colorlinks=true,linkcolor=blue,citecolor=red]{hyperref}
\usepackage[svgnames,dvipsnames]{xcolor}
\usepackage{tikz}
\tikzset{%
  from end of path/.style={
    insert path={
      \pgfextra{%
        \expandafter\pgfprocesspathextractpoints%
          \csname tikz@intersect@path@name@#1\endcsname%
        \pgfpointlastonpath%
        \pgfgetlastxy\lastx\lasty
      }
      (\lastx,\lasty)
}}}

\newtheorem{theo}{Theorem}[section]
\newtheorem{lem}[theo]{Lemma}
\newtheorem{prop}[theo]{Proposition}

\newtheorem{coro}[theo]{Corollary}
\newtheorem{rema}[theo]{Remark}

\newcommand{\M}{\mathcal{M}}

\newcommand{\DD}{\mathcal{D}}
\newcommand{\vr}{\bm{r}}
\newcommand{\K}{K_{\Lambda}}

\newcommand{\rx}{\varphi_{\Lambda}\{a_m\}_{m=1}^{\infty}}
\newcommand{\vq}{\varphi_{\Lambda}\{a^q_m\}}
\newcommand{\vl}{\varphi_{\Lambda}\{a^l_m\}}

\newcommand{\vv}{\bm{v}}

\newcommand{\vf}{\bm{\varphi}}
\newcommand{\ro}{\rho_{\Lambda}}
\newcommand{\Ro}{\bm{\rho_{\Lambda}}}
\newcommand{\si}{\underset{z \to z_c}{\sim}}
\let\a=\alpha \let\b=\beta   \let\e=\varepsilon
       \let\l=\lambda
 \let\n=\nu     \let\p=\pi \let\ph=\varphi
\let\r=\rho   
 \let\x=\xi 
\let\D=\Delta \let\F=\Phi  \let\L=\Lambda 
    \let\X=\Xi

\title{On Kirkwood-Salsburg solutions at criticality  }

\author{Rogério Gomes Alves }
\date{}

\begin{document}

\maketitle

\begin{abstract}

 In this work we study the Kirkwood-Salsburg equations of equilibrium classical continuous systems.  We prove a Laurent expansion for the resolvent, at an eigenvalue of largest modulus of the Kirkwood-Salsburg  operator,   which  is shown to have a pole of order 1. Then we prove that all  correlation functions have an  asymptotic limit as the activity parameter tends to a  smallest zero of the partition function.  As  corollary, we show that  any smallest zero  of the partition function  is simple. The main consequence    is that in case of positive or hardcore potentials we  find the spectral radius of the Kirkwood-Salsburg operator and the convergence radius of the solutions.  
\end{abstract}

\section{Introduction}

  The Kirkwood-Salsburg equations, or KS equations for short, is an infinite system of linear inhomogeneous integral equations  satisfied by the correlation functions of classical continuous systems in statistical  mechanics, see the books Hill \cite{H}, Ruelle \cite{Ru} and   Gallavotti \cite{Ga}. These systems of equations have been developed and studied over the years in statistical mechanics almost always to  study phase transitions. In fact, these equations seems to have its origins in the work of Boltzmann about hard sphere fluid, and its modern form is due to Mayer, Kirkwood and their collaborators, see \cite{H}. The KS equations appears in others contexts in statistical mechanics. For instance, in Brascamp \cite{B}  is shown that the KS equations for a classical lattice gas are equivalent to the Dobrushin-Lanford-Ruelle equations. Applications of the KS equations on convergence problems of polymer models can be find in Bissacot  et al \cite{BFP}  (and references therein)  and on point process in Kuna et al \cite{KLS}.     In this work we explore the spectral properties of the parameter activity $z \in \mathbb{C}$ of the Kirkwood-Salsburg equation to study the  behavior of the correlation functions near a leading singularity.
  
    The KS operator has non trivial spectral properties which deeply depends on the operator domain as well as the interacting potential. We may see this in Ruelle \cite{Ru1} and  \cite{Ru},  Pastur \cite{Pas}, Gorzela\'{n}czyk \cite{Go} and Zagrebnov \cite{Z} and references therein,.
 
 In \cite{Ru}, the Kirkwood-Salsburg equations were introduced as an equation in a  Banach space in which  is defined the Kirkwood-Salzburg operator. This opened a new way to study classical correlation functions with relevant results.     The author in \cite{Pas}  explore the fact that the correlation functions in finite volume are given by the ratio of two entire functions, with the  partition function in the denominator,  to built a Fredholm Theory for the  Kirkwood-Salsburg system in an wider Banach space. The main result in \cite{Pas} is that the spectrum of the KS operator coincides with the set $z_l^{-1},\, l=1,2,\dots$, which $z_l$ are zeros of the partition function. In particular, we have only the point spectrum for the KS operator. This same result is find in  \cite{Go} in a smaller Banach space. In \cite{Go} the author also shows a bound on the spectrum of the KS operator in case of  Hard Core potentials.  Then for such systems it is proved in \cite{Go} that all correlation functions are analytic for nonnegative activity $z$ such that $ z \le [(e-1)C]^{-1}$, that is, phase transition does not occur for these activities values. The positive constant $C$ is defined in the next section in (\ref{constante C(beta)}). From this  the system is still in the gas phase if the density $\r_1$ is such that $\r_1<[eC]^{-1}$.  This result expanded the previous known region   $\r_1 < [(e+1)C]^{-1}$, which was proved  in \cite{Ru}  for nonnegative potentials.  In \cite{Z},   the author was concern about the spectral properties of the KS operator (defined with nonempty boundary conditions, but the results there follows as well to empty boundary conditions) considering different potentials and domains. He  proves, among other things, that the KS operator  have  only the point spectrum, in which is equal to the inverse of zero's partition function, if we take the domain of the KS operator as the  Banach subspace of symmetric sequences and  a interacting potential with repulsive part as  a hard core or positive at least in a neighborhood of the origin.

In this work we assume that the interacting potential and the domain of the KS operator,  are such that it has the property to have only the point spectrum in which is equal to the set of inverse zeros of the partition function.  This a important property to require since, besides we find it  in literature on relevant situations (e. g. positive and hard core potentials) as cited above, we can explore the functional analysis.   We show, for all volume $\L$, that  there exists a Laurent expansion of the resolvent operator, for the  KS operator, at an eigenvalue $\l_c=z_c^{-1}$ such that its modulus is equal to the spectral radius of the KS operator. Then we prove that $\l_c$ is a pole of order $1$ and that all $n-$point $\L$ correlation functions are asymptotically to $\displaystyle \frac{M_n z}{1-z/z_c}$ as  $z \to z_c$, for some nonzero constants $M_n=M_n(\L)$.   From these and the result that inverse of spectral values are zeros of the partition function, mentioned in the preceding paragraph, we have that the smallest zero of the partition function (the closest to the origin)  $z_c$ is simple. The main consequence of our tools is the prove, for stable and regular potentials, that the convergence radius of the Cluster expansion is at least $C^{-1}$, with equality if there is a zero of the partition function with modulus equal to $C^{-1}$.   For positive or hard core potentials, we then have that  the convergence radius of the cluster expansion is equal to $C^{-1}$, with singularity at $z_c=-C^{-1}$.  This solves an old problem in classical statistical mechanics of  equilibrium of determining the convergence radius of the  Cluster expansion of the gas density, as may be seen in the classical works of Groeneveld \cite{Gr}, Penrose \cite{P1}, Ruelle \cite{Ru1} and \cite{Ru}, as well as in Gorzela\'{n}czyk \cite{Go} and \cite{Go 0}.  We also prove that the convergence radius of the Virial expansion for positive or hard cores potentials is at least $(2C)^{-1}$, which improves the results in Ruelle \cite{Ru} and Gorzela\'{n}czyk \cite{Go}.
 
   Our method to study  correlation functions at criticality are motived from combinatorial problems in which complex analysis play an important role as in  analytic combinatorics, see book Flajolet and Sedgewick \cite{FO}. In order to find the singularity of the correlation functions, we start looking for an equation involving them (KS equations), following some standard techniques to study singularities of a system of generating functions. The problem is that in case of an infinite system of correlation functions a lot of  other questions arise and this approach is not direct.  We hope that these   ideas can be applied in  other  problems. 
   
   The next Subsection \ref{classical systems} contains the definitions. The results and their proofs  are in the Section \ref{main results}. We also include an Appendix, Section \ref{apendice}, and a short discussion of the results in the last Section \ref{dicussao}.
   
%%%%%%%%%%%%%%%%%%%%%%%%%%%%%%%%%%%%%%%%%%%%%%%%%%%%%%%%%%%%%%%%%%%%%%%%%

\subsection{Classical continuous systems}\label{classical systems}

 Consider a classical continuous system of identical point particles in a bounded volume $\L \subset \mathbb{R}^{\n}$. The  \textit{$n$-point correlation functions}  in the volume $\L$ with activity $z \in \mathbb{C}$, in the grand canonical ensemble formalism,  are defined  by
 
 \begin{equation}\label{correlation func}
\ro(x)_n= \ro(z;(x)_n)=\X_{\L}^{-1}(z)\chi_\L(x)_n \sum_{m=0}^{\infty} \frac{z^{n+m}}{m!} \int_{\L^m} d(y)_m  e^{-\b U((x)_n,(y)_m)},
  \end{equation}
where $(x)_n=(x_1,\dots,x_n) \in \mathbb{R}^{\n n}$,   $\int_{\L^m} d(y)_m=\int_{\L} dy_1 \dots \int_{\L} dy_m $,
\begin{equation*}
\chi_\L(x)_n=\prod_{j=1}^{n}\chi_\L(x_j), \,\,\,\,\, \chi_\L(x)=  \begin{cases} 1, &  x \in \L \\ 0, & x \notin \L  \end{cases}  
\end{equation*} 
 and   $\X_{\L}(z)$ is the grand-canonical partition function given by
 
 \begin{equation}\label{part func}
 \X_{\L}(z)=1 + \sum_{m=1}^{\infty}{z^m\over m!} \int_{\L^m} d(y)_m  e^{-\b U(y)_m}.
  \end{equation} 
 
 The pairwise  interaction potential $\F$ is a function depend on the distance between particles and the energy of $n$ particles is the function $(x)_n \in \mathbb{R}^{\n n} \mapsto U(x)_n=\sum_{i<j} \Phi(|x_i-x_j|) \in \mathbb{R}$   which is   symmetric respect with all permutations of their arguments and is  \textit{stable}, that is,   for all $n \ge 1 $ and $x_1,\dots,x_n \in \L$ there is a constant $B \ge 0$ such that 
 \begin{equation}
 U(x)_n \ge -Bn.
 \end{equation}
  We also assume the \textit{regularity} of the potential,   
 
 \begin{equation}\label{constante C(beta)}
 C:=C(\b)=\int_{\mathbb{R}^{\n}} |e^{-\b\Phi(x)} -1| dx < \infty.
 \end{equation}

The correlation functions satisfy an infinite system of integral \textit{Kirkwood-Salsburg equations}  \cite{Ru}:

\begin{equation}\label{KS 0}
\ro(z;x_1)=z \chi_\L(x)_1 \big(1+ (K_\L \ro)(x_1) \big) ,
\end{equation}

\begin{equation}\label{KS 01}
\ro(z;(x)_n)= z\chi_\L(x)_n(K_\L \ro )(x)_n,
	\end{equation}
where $K_\L$ is the \textit{Kirkwood-Salsburg operator}:

\begin{equation}
(K_\L \ro )(x_1)=\sum_{m=0}^{\infty} {1 \over m!} \int_{\L^m} d(y)_m  K(x_1;(y)_m) \ro (y)_m, 
\end{equation}	
and 
\begin{equation}
(K_\L \ro )(x)_n=e^{-W(x_1;(x)'_n)} \sum_{m=0}^{\infty} {1 \over m!} \int_{\L^m} d(y)_m K(x_1;(y)_m) \ro ((x)'_n,(y)_m)
\end{equation}	
with  $(x)'_n:= (x_2,\dots,x_n)$,  $W(x_1;(x)'_n):=U(x)_n - U(x)'_n$ and $$ K(x_1;(y)_m):= \prod_{k=1}^{m}(e^{-\b \Phi(|x_1-y_k|)}-1).$$

  The Kirkwood-Salsburg equation is considered in \cite{Ru} in  the Banach space $\mathcal{E}_\x(\L)$, with respect to the norm $\parallel \cdot \parallel_{\x}$,
  
  \begin{equation}\label{banach space}
  \mathcal{E}_\x(\L) =\{\vf=(\ph(x)_n)_{n \ge 1}:\,\,\,  \parallel \vf \parallel_{\x} = \sup_{n \ge 1} \x^{-n} \mbox{ ess}\sup_{(x)_n \in \mathbb{R}^\n} |\ph(x)_n| < \infty\}.
  \end{equation} 
 We may see that the correlation functions belongs to $\mathcal{E}_\x(\L)$ if $|z|< \x e^{-\b B}$.

 Note that $\mathcal{E}_{\x_2}(\L) \subset \mathcal{E}_{\x_1}(\L)$, if $\x_1>\x_2$. The KS operator may be defined as $K_\L: \mathcal{E}_\x(\L) \to \mathcal{E}_{e^{2\b B}\x(\L)}$, see \cite{Ru}.

  In order to get spectral results we will consider the KS operator defined on a Banach subspace $\DD_\x(\L) = \mbox{ domain}(\K) \subseteq \mathcal{E}_\x(\L)$ such that  $\K:\DD_\x(\L)  \to \DD_\x(\L)$ has the \textit{property} that its spectrum consists of only the point spectrum in which is equal to the set of inverse  zeros of the partition function. This is not a strong assumption, as we can see below, since we find it in relevant situations.

    There are, as we know, two special closed subspace $\DD_\x(\L)$ with the desired property above:
  
  (i) $\DD_\x(\L)=\mathcal{E}^s_\x(\L)$, in which $\mathcal{E}^s_\x(\L)$ is the subspace of symmetric sequences.
  
  (ii) $\DD_\x(\L)=D_\x(\L)$, in which 
  
   \begin{equation*}
  D_\x(\L)=\{\vf  =\big(\varphi_{\Lambda}\{a_m\}(x)_n\big)_{n \ge 1} \in \mathcal{E}_\x(\L):\,\,  \varphi_{\Lambda}\{a_m\}(x)_n =\sum_{m=0}^{\infty} \frac{a_{n+m}}{m!} \int_{\L^m} d(y)_m  e^{-\b U((x)_n,(y)_m)}< \infty 
  \end{equation*}
  \begin{equation}
  \forall n \ge 1 , \forall (x)_n \in \mathbb{R}^{n\n}, \mbox{ and a complex sequence } \{a_m\}_{m=1}^{\infty} \}.
  \end{equation}
  
   The subspace in (i) is the natural to carry symmetries, \cite{Ru1},  \cite{Z} and the subspace in (ii) appears in \cite{Go}. We give a proof in the Appendix, Section \ref{apendice}, that the subspace in (ii) is a Banach space. This subspace seems to be a natural choice, since besides in \cite{Go}, the Banach subspace in \cite{Pas}, see Lemma 2 therein, may be thought as some $D_\x$.

In the following,  we state  the  Theorems (\ref{zagrebnov}) and (\ref{Pastur}) below, in which we have the property for the KS operator to have only the point spectrum equal to the set of inverse zeros of the partition function.
  
  \begin{theo}\label{zagrebnov}
  	(Zagrebnov \cite{Z}, Theorem 3.5) If $\DD_\x(\L) = \mathcal{E}^s_\x(\L)$ and the  pair potential is regular and its repulsive part has a hard core or is positive in some neighborhood of the origin, then  the spectrum of $K_\L$ on $\DD_\x(\L)$ consists of only the point spectrum, and a complex number $\l \not= 0$ belongs to the spectrum  if and only if $z=\l^{-1}$ is a zero of the partition function.
  \end{theo}

  \begin{theo}\label{Pastur}
  	( Gorzela\'nczyk \cite{Go}, Proposition 1; Pastur \cite{Pas}, Theorem ) 
  	If the pair  potential is regular and stable, then  the spectrum of $K_\L$ on $\DD_\x(\L)$ consists of only the point spectrum, and a complex number $\l \not= 0$ belongs to the spectrum of $K_\L$ on $\DD_\x(\L)$ if and only if $z=\l^{-1}$ is a zero of the partition function. The multiplicity of the eigenvalue $\l$ is equal to the corresponding zero.  In  \cite{Go} we have $\DD_\x(\L) = D_\x(\L)$ and in \cite{Pas} we have other subspace $\DD_\x(\L)= \DD \subseteq \cup_{ \x>0}  \mathcal{E}^s_\x(\L)$.
  \end{theo}

  \begin{rema}\label{obs 1}
  Our results are based on the fact that the KS operator $\K$ has the property of have only the point spectrum equal to the set of inverse zeros of the partition function. So, what we need is a Banach subspace $\DD_\x(\L)$ and a suitable potential interaction to have this. Although we work with this property itself, the achievement of our results are conditioning on the more general setting in which we have such property.
  
  From the theorems above, we  can consider the KS operator $\K$ with our property  in one of the followings situations:

  	(a)  $\DD_\x(\L) = \mathcal{E}^s_\x(\L)$ and the pair potential  is regular and its repulsive part has either a hard core or is positive in some neighborhood of the origin.
  	
  	(b)  $\DD_\x(\L) = D_\x(\L)$ and pair potential  is regular and stable.

  	In front of this, we consider the KS operator $\K:\DD_\x(\L) \to \DD_\x(\L)$, which is bounded if it is defined following  either (a) or (b) above, see Appendix, Proposition \ref{prop 3 apend}.  We will state our theorems under the assumptions in (b) in order to obtain a more general result.
   \end{rema}

 \begin{rema}\label{obs 2}
 	The subspaces  $\DD_\x(\L) = D_\x(\L)$ and $\DD_\x(\L) = \mathcal{E}^s_\x(\L)$ are Banach spaces invariant under the action of $\K$, see \cite{Go}, \cite{Z} or Appendix Proposition \ref{prop 2 apend}. In fact, the Kirkwood-Salsburg operator $\K$ acting on $\ph(x)_n=\rx$ defines an operator (also denoted by $\K$) in the space of coefficients$\{a_m\}_{m=1}^{\infty}$ acting as
 	
 	\begin{equation}\label{eq remark}
 	\K\{a_m\}_{m=1}^{\infty}=\{a_{m-1}\}_{m=1}^{\infty}, \,\,\,\,\,\,\,\,\,\,\,\,\, a_0=-\sum_{m=1}^{\infty} \frac{a_{m}}{m!} \int_{\L^m} d(y)_m  e^{-\b U((y)_m)}.
 	\end{equation}
 	This is already known important fact since from \cite{Pas}, which may be seen  by directed calculations.  Due this action of the KS operator, which is similar to the right shift, we have  for the particular case of the ideal gas, that the spectrum of the KS operator consists of only the residual spectrum. See \cite{Z} for a proof and a more involved discussion. 
 \end{rema}

%%%%%%%%%%%%%%%%%%%  equação principal %%%%%%%%%%%%%%%%%%%%%%%%%%%%%

We want to consider the solvability  of the   Kirkwood-Salsburg (KS) equations on a  $\DD_\x(\L)$, such that the KS operator has only the point spectrum in which coincides with the inverse of zeros of the partition function,

  \begin{equation}\label{KS 2}
(I-z\K)  \vf =z\bm{\a}_\L  ,
  \end{equation}
or 
 \begin{equation}\label{KS 3}
 (\K - \l)\vf =-\bm{\a}_\L,
 \end{equation}
 with $\l=z^{-1}$, $\vf \in \DD_\x(\L)$, and $\bm{\a}_\L=(\a_\L(z;(x)_n))_{n \ge 1}$, where $\a_\L(z;(x)_n)=1$, if $n=1$ and $0$ otherwise.

%%%%%%%%%%%%%%%%%%%%%%%%%%%%%%%%%%%%%%%%%%%%%%%%%%%%%%%%%%%%%%%%%%%%%%

 We will consider the \textit{resolvent operator} $R(\l):=(\K-\l)^{-1}$ \footnote{We consider the resolvent operator as in \cite{K} instead of $(\l-\K)^{-1}$. } as an operator valued  function of one complex variable.

   By the symbol $f \si g$ we mean that  the functions $f,g:\D \to \mathbb{C}$,  for $\D \subset \mathbb{C} \setminus \{z_c\}$, $z_c \in \bar{\D}$, satisfy $\displaystyle \lim_{z \to z_c} f(z)/ g(z) = 1$.  We also use the same notation for sequences of complex numbers $f_n \sim g_n$ to mean $\displaystyle \lim_{n \to \infty} f_n/ g_n = 1$.

 %%%%%%%%%%%%%%%%%%%%%%%%%%%%%%%%%%%%%%%%%%%%%%%%%%%%%%%%%%%%%%%%%%%%%%%%%%%%%%%%%%%%%%%%%

\section{Results: critical behavior}\label{main results}

From now on we assume that the Kirkwood-Salsburg operator are defined, for a non ideal gas, as in (b). See Remarks \ref{obs 1} and \ref{obs 2}.

The goal is to explore the spectral properties of the Kirkwood-Salsburg operator in the problem of solving the equation (\ref{KS 2})-(\ref{KS 3}).    It can be seen along this work that our results comes from the fact that the critical behavior of correlation functions are determined by its components in the direction corresponding to a leading simple eigenvalue in a spectral decomposition of the solution of the Kirkwood-Salsburg equation. We shall see that we are able to determine the  behavior of the solution of this equation near a smallest zero of the partition function.

       The Theorem \ref{teorema principal} below gives the asymptotic behavior of the correlation functions  and its  coefficients at a leading singularity. As a Corollary (\ref{cor principal}), we have that a smallest zero of the partition function is simple.

       \begin{theo}\label{teorema principal}
          Let be  $\l_c=z_c^{-1} \not=0$ in the spectrum of  $\K$.  Then, for all $\L$, in a neighborhood of $\l_c$ holds,   
          
          \begin{equation}\label{resolvent exp}
          R(\l)={-P \over \l-\l_c} + \sum_{n \ge 0}(\l-\l_c)^n S^{n+1}, 
          \end{equation}     
         with
         \begin{equation}\label{projection}
         P:={-1 \over 2\pi i}\int_{\mathcal{C}} R(\l)\, d\l \: \: \:\:\:\:\:\: S:={1 \over 2\pi i}\int_{\mathcal{C}} {R(\l) \over \l -\l_c}\, d\l,
         \end{equation}
         where $\mathcal{C}$ is a circle centered at $\l_c$ which its interior, except $\l_c$, and the own $\mathcal{C}$ is contained in  the resolvent set of $\K$. Then  the resolvent operator $R(\l):=(\K-\l)^{-1}$ has a pole of order 1 at $\l_c$ and it is an eigenvalue of multiplicity  $1$.

          Consider the complex function $f(z):=\vv_c^*(P\Ro(z))$, defined on a domain $\D \subset \mathbb{C} \setminus \{z_c\}$ to be given, where $\vv_c^*$ is the dual element of the eigenvector $\vv_c$ associate with $\l_c$. Then
          
       \begin{equation}\label{asymp}
       f(z) =  {Mz \over 1 - z/z_c},  
       \end{equation}
       for  some nonzero constant $M=\vv_c^*(\a_\L)$. For all $n \ge 1$,
       
       \begin{equation}\label{asymp 1}
       \ro(z,(x)_n) \si  {M_nz \over 1 - z/z_c},  
       \end{equation}
       for  some nonzero constants $M_n=M_n(\L)$, restricting  $\ro(z,(x)_n)$  on  $\D$. Consequently   $|\l_c|=r(\K)$, the spectral radius of $\K$.              	               	   
           \end{theo}

      \begin{coro}\label{cor principal}
     The smallest zero $z_c=\l_c^{-1}$ of the partition function \ref{part func} is simple.   
      \end{coro}

        \begin{coro}\label{cor principal 1}
        	The $k$-coefficient of the Taylor expansion of $\ro(z;(x)_n)$ around the origin satisfies  for all $n \ge 1$,
        	
        	\begin{equation}
        	[z^k]\frac{\ro(z;(x)_n)}{z} \sim M_n\l_c^k + o(\l_c^{-k}),
        	\end{equation}
        	as $k \to \infty$.
                \end{coro}

       The last corollary stress that its asymptotic result comes from the fact that we have a pole of order $1$. A pole of order $d>1$, for instance, would imply a subexponential factor   of the form $k^{d-1}$, that is, we would have $\sim \l_c^k k^{d-1}$.

  The theorems (\ref{teorema principal 2}) and (\ref{cor principal 2}) determines respectively, for all $\L$, a bound for the spectral radius $r(\K)$  of the Kirkwood-Salsburg operator  and the gas phase region for systems with stable and regular pairwise  interacting  potential. 
      
       \begin{theo}\label{teorema principal 2}
       	   Let be $r(\K)$ the spectral radius  in $\DD_\x(\L)$ . Then, for all $\L$,
       	
       	\begin{equation}\label{raio HC}
       	r(\K) \le \x^{-1}.
       	\end{equation} 
       \end{theo}

      \begin{theo}\label{cor principal 2}
      Consider a nonzero pair potential stable and regular.  If $\L \nearrow \mathbb{R}^\n$ is an increasing sequence, then for all $n \ge 1$ the limit 
      	
      	\begin{equation}
      	 \r_n(z)=\lim_{\L \nearrow \mathbb{R}^\n} \ro(z;(x)_n),
      	\end{equation}
      	there exists for all $z$ in the disc $|z| <C^{-1}$. All correlation functions  $\r_n(z)$  have convergence radius  at least $C^{-1}$. If there is a zero of the partition function with modulus $C^{-1}$, then the convergence radius  is $C^{-1}$.   
      \end{theo}

       \begin{coro}\label{cor principal 3}
       	For positive or hard core potentials, the convergence radius of $\r_1(z)$ is $C^{-1}$,  with the singularity at $-C^{-1}$.  The convergence radius of the Virial expansion is at least $(2C)^{-1}$. 
       \end{coro}

 Before start with the main demonstrations, we  give two propositions. The first one is a straightforward consequence of definitions and from the equations (\ref{KS 0})-(\ref{KS 01}). 
 
 \begin{prop}\label{prop 1}
 	The vector $\displaystyle \ro(z;(x)_n)=\X_\L(z)^{-1} \varphi_\L \{z^m\}(x)_n$, $n \ge 1$,   
 	
 	\begin{equation}\label{eq proposicao}
 	 \varphi_\L \{z^m\}(x)_n = \sum_{m=0}^{\infty} \frac{z^{n+m}}{m!} \int_{\L^m} d(y)_m  e^{-\b U((x)_n,(y)_m)},
 	  	\end{equation}
 	  	satisfies formally the KS equations (\ref{KS 2})-(\ref{KS 3}).
 \end{prop}

 Remember that the resolvent operator $R(\l)=(\K-\l)^{-1}$ is unbounded if and only if $\l=z^{-1}$ belongs to the spectrum of $K_\L$. So, from  Theorems \ref{zagrebnov} and \ref{Pastur}, we have the following proposition.

\begin{prop}\label{prop 2}
	The complex number $z_c=\l_c^{-1} \not= 0$ is a zero of the partition function if and only if the operator $R(\l_c)$ is unbounded  in $\DD_\x(\L)$.
\end{prop}

From these propositions we have that the correlation functions $ \Ro(z)$  are in $\DD_\x(\L)$ if and only if $z^{-1}$ is in the resolvent set of $\K$.

%%%%%%%%%%%%%%%%%%%%%%%%%%%

%%%%%%%%%%%%%%%%%%%%%%%%%%%%%%%%%%%%%%%%%%%%%%%%%%%%%%%

\subsection{Proof of the Theorem \ref{teorema principal}}

\begin{proof}
\textit{Theorem \ref{teorema principal}}.

Let be $\l_c=z_c^{-1} \not= 0$ a complex number in the spectrum of   $K_\L$. 

First, note that the projections operators $P$ and $I-P$ are well defined because $\l_c=z_c^{-1}$ is an isolated point of the hole spectrum, since we are assuming that the KS operator has only the point spectrum which is equal to the set of inverse zeros of the partition function.    Therefore there exist a Laurent  expansion for the resolvent  as

\begin{equation}\label{resolvent exp2}
R(\l)= {-P \over \l-\l_c} - \sum_{n=1}^{\infty} {D^n \over (\l-\l_c)^{n+1}} +  \sum_{n \ge 0}(\l-\l_c)^n S^{n+1},  
\end{equation}
where $P$ and $S$ were defined and $D:={-1 \over 2\pi i}\int_{\mathcal{C}} (\l-\l_c)R(\l)\, d\l $. Both  projections $P$ and $I-P$ commutate with $\K$, and holds $PD=DP=D$, $\K S=S \K$ and $SP=PS=\bm{0}$ . See book Kato \cite{K} Chapter 3, Section 6  for  all these and more properties related with these projections. We will state some others as we need. We can see similar results in the book Dunford and Schwartz \cite{DS} Chapter VII.

      The series involving the operator $D$ in (\ref{resolvent exp2}) could be either an infinite series or an  expansion with only $n_c$ finite terms  if the subspace $\M_c=P(\DD_\x(\L))$ is finite dimensional $n_c=\text{dim}\M_c<\infty$. We want to show that $D=\bm{0}$ and then the  series  in the expansion  (\ref{resolvent exp2}) involving the operator $D$ vanishes.  To see this, first we note that assembling the powers of $\l-\l_c$ in $(\K-\l)R(\l)=I$, we will find for all $n=1,2, \dots$,
      \begin{equation}\label{eq para D}
      D^n[\K(\l_c) -D ]=\bm{0},
      \end{equation}
      with $\K(\l_c):=\K -\l_c$. Indeed, we have
      
      \begin{equation}\label{eq 0 para D}
      I=(\K-\l)R(\l)=  \big[\K(\l_c) -(\l-\l_c)\big] \Big[{-P \over (\l-\l_c)} - \sum_{n=1}^{\infty} {D^n \over (\l-\l_c)^{n+1}}\Big] + (\K-\l)S\sum_{n \ge 0}(\l-\l_c)^n S^{n}.
      \end{equation}
      With the property that $\K(\l_c) P=(\K-\l_c)P=P(\K-\l_c)=D$ and writing as one series the first term in the sum of r.h.s. of (\ref{eq 0 para D}), we have
      
      \begin{equation}\label{eq 1 para D}
      I=(\K-\l)R(\l)= P + \sum_{n=1}^{\infty} {D^n[\K(\l_c)-D] \over (\l-\l_c)^{n+1}} + (\K-\l)S\sum_{n \ge 0}(\l-\l_c)^n S^{n}.
      \end{equation} 
      Now, since hold the identities $S=(I-P)S=S(I-P)$, and $(\K-\l_c)S=S(\K-\l_c)=I-P$, see \cite{K}, we have that 
      \begin{equation*}
      (\K-\l)S=\big[\K-\l_c -(\l-\l_c)\big]S=I-P - (\l-\l_c)S =
      \end{equation*}
      \begin{equation}\label{eq 2 para D}
      =(I-P)[I - (\l-\l_c)S].
      \end{equation}
      Then, putting the identity  (\ref{eq 2 para D}) in the last series of r.h.s. (\ref{eq 1 para D}) we end up with
      \begin{equation}
      (\K-\l)S\sum_{n \ge 0}(\l-\l_c)^n S^{n} = I-P,
      \end{equation}
      which together with (\ref{eq 1 para D}) gives (\ref{eq para D}).

      We claim that the equation (\ref{eq para D}) implies $D=\bm{0}$. Indeed, either $D^n=\bm{0}$ for all $n=1,2,\dots$, or we must have $\bm{0}=\K(\l_c) -D= \K - \l_c -D$, that is, $D=\K-\l_c$. Thus,  for all $\vr \in \M_c$ we have $D\vr=\bm{0}$ and then $D\mid_{\M_c}=\bm{0}$. For $\M'_c=(I-P)(\DD_\x(\L))$,  $\vr=(I-P)\vv \in \M'_c,\, \vv \in \DD_\x(\L)$, note that $D\vr=D(I-P)\vv=(D-DP)\vv=(D-D)\vv=\bm{0}$, since holds the identity $DP=PD=D$. From this we have that $D\mid_{\M'_c}=\bm{0}$ and  then,   we obtain  $D=DP +D(I-P) = D\mid_{\M_c} + D\mid_{\M'_c} =\bm{0}$  on $\DD_\x(\L)=\M_c \oplus \M'_c$.  Thus  the part involving $D$ in the expansion  (\ref{resolvent exp2})  vanishes, and since $D^1=\bm{0}$  we have $\text{dim}\M_c=1$.  Then  $\l_c$ is an eigenvalue of (algebraic) multiplicity $1$,  since it belongs to the spectrum of $\K$ restricted to $\M_c$, which is one dimensional operator, so it is an eigenvalue of the restricted operator and of  $\K$ as well.    Therefore is showed the expansion (\ref{resolvent exp}) and $R(\l)$ has a pole of order $1$ at the eigenvalue $\l_c$.
 
 Now we will prove the claims (\ref{asymp}) and (\ref{asymp 1}). 
 
 In order to get the asymptotic result  we need  the Laurent expansion of $\Ro$ at $\l_c$. From the expansion for the resolvent (\ref{resolvent exp}), we get
 
 \begin{equation}\label{eq Teo}
 \Ro=-R(\l)\bm{\a}_\L={P \bm{\a}_\L \over \l-\l_c}  -  \sum_{n \ge 0}(\l-\l_c)^n S^{n+1}\bm{\a}_\L.
 \end{equation}

The subspace $\M_c$ is one dimensional and so, le be $\vv_c$ a basis for $\M_c$. Because of that and by the fact we have an isolated eigenvalue, the subspace $\M_c^*$, which is the range of $P^*$, is also one dimensional, see for instance Kato \cite{K} Chapter 3, Section 6 for details. Then there is a well defined basis $\vv_c^*$ of $\M_c^*$, such that $\vv_c^*(\vv_c)=1$ and for all $\vv \in \DD_\x(\L)$ we may write $P\bm{\vv}=\vv_c^*(\vv)\vv_c$. Since holds the identity $SP=PS=\bm{0}$, applying $P$ in the equation (\ref{eq Teo}), we have 
\begin{equation}\label{eq teo princ}
P\Ro={P \bm{\a}_\L \over \l-\l_c}={\vv_c^*(\bm{\a}_\L) \over \l-\l_c} \vv_c.
\end{equation}
Let be  $ \Delta:=\{z: |z -z_c| < \eta \mbox{ and } |arg(z - z_c)| > \theta\} \subset \mathbb{C} \setminus \{z_c\}$, for some $\eta>0$ and $0<\theta<\p/2$ such that it is free of other zeros of the partition function, that is, we want that the image of $\D$ by the map $0\not= z \mapsto z^{-1}$ is contained  in the interior of the circle $\mathcal{C}$ except $\l_c$. Thus (writing $z=\l^{-1}$)  we define on $\D$
\begin{equation}\label{eq Teo1}
 f(z):=\vv_c^*(P\Ro(z))={\vv_c^*(\bm{\a}_\L) \over \l-\l_c}={z\vv_c^*(\bm{\a}_\L)  \over 1-z/z_c},
\end{equation}
which proves (\ref{asymp}) with $M=\vv_c^*(\bm{\a}_\L)$.
    
    We also have $(I-P)S=S(I-P)=S$, so we rewrite the equation (\ref{eq Teo}) as 
    \begin{equation}\label{eq Teo 2}
    \Ro=P\Ro + (I-P)\Ro,
    \end{equation}
    and  the analytic part of $\Ro$ in the expansion (\ref{resolvent exp}) is
    \begin{equation}
    (I-P)\Ro=-\sum_{n \ge 0}(\l-\l_c)^n S^{n+1}.
    \end{equation} 
   Then,
    \begin{equation}\label{eq Teo 4}
     \Big({1-z/z_c  \over z}\Big)(I-P)\Ro(z) =(\l-\l_c)(I-P)\Ro(z)\to \bm{0},
    \end{equation}
       as $z \to z_c$ in $\D$. Therefore, follows from (\ref{eq teo princ})-(\ref{eq Teo 4}) that 
       \begin{equation}\label{eq Teo 3}
       \ro(z;(x)_n)={M_n z  \over 1-z/z_c} + o\Big({z  \over 1-z/z_c}\Big),
       \end{equation}
       as $z \to z_c$ in $ \Delta$, for some constants $M_n=M_n(\L) \neq 0$. This proves (\ref{asymp 1}). Finally, from this and since  $|z_c|=|\l_c|^{-1}$, we have $|\l_c|\ge |\l|$ for any $\l$ in the spectrum of $\K$ and so $|\l_c|=r(\K)$,  the spectral radius of $\K$.

\end{proof}

\begin{proof} Proof of the Corollary \ref{cor principal}.
	
	The proof follows from the theorem's result and  Theorem (\ref{Pastur}).
\end{proof}
\vspace{2mm}

\begin{proof} Proof of the Corollary \ref{cor principal 1}.

	To prove this corollary we use  the Theorem VI.4 in \cite{FO} which translate the asymptotic of the correlation function in (\ref{asymp 1}) to coefficients. For that, since $\D$ is free from other zeros of the partition function, the identity (\ref{eq Teo 3})  implies that  $\ro(z;(x)_n)/z$ have an analytic continuation to the region $\D$, and 
	\begin{equation}
	\frac{\ro(z;(x)_n)}{z}={M_n  \over 1-z/z_c} + o\Big({1  \over 1-z/z_c}\Big).
	\end{equation}
	 So applying the theorem cited above  for $\ro(z;(x)_n)/z$, in its little-$o$ version, we have the desired result.
\end{proof}

%%%%%%%%%%%%%%%%%%%%%%%%%%%%%%%%%%%%%%%%%%%%%%%%%%%%%%%%%%%%%%%%%%%%%%%%%
\subsection{Proofs of the Theorems \ref{teorema principal 2} and \ref{cor principal 2}}

 We denote by  $\l_c$ the unique eigenvalue of $\K$ such that $|\l_c|=r(\K)>0$.    

We start with the Lemma \ref{lema2} below. The proofs of Theorems (\ref{teorema principal 2}) and \ref{cor principal 2} are possible through this key lemma, in which comes from the Theorem \ref{teorema principal}.

\begin{lem}\label{lema2}
	The Kirkwood-Salsburg operator can be written $\DD_\x(\L)$ as  
	\begin{equation}\label{decomp K 1}
	\K=\l_cP +T,
	\end{equation}
	which $T$ is an operator that satisfies $PT=TP=\bm{0}$, $\l_c=r(\K)$. As a consequence we have 
	\begin{equation}\label{decomp K 2}
	(\l_c^{-1}\K)^n \to P,
	\end{equation}
	 as $n \to \infty$, in the operator norm.
\end{lem}

\begin{proof}
	From the proof of the Theorem \ref{teorema principal}, applying $\K$ in the equation (\ref{eq Teo 2}) we have (\ref{decomp K 1}) with $T=\K(I-P)$. Remember that $\K$ also commutes with $P$.  Since $P^2=P$, we have $PT=TP=\bm{0}$ and so for all $n \in \mathbb{N}$ $\K^n=\l_c^nP +T^n$ or
	
	\begin{equation}\label{eq lema2}
(\l_c^{-1}\K)^n -P = (\l_c^{-1}T)^n, \,\,\,\, \forall n \in \mathbb{N}.
	\end{equation} 
The righ hand side (\ref{eq lema2}) above tends to zero in the operator norm. In fact,  the  operator  $T$ is the KS operator  restricted to $(I-P)\DD_\x(\L)$ and  $\l_c$ is the unique eigenvalue of $\K$ of largest modulus. Thus we have that the spectral radius of  $\l_c^{-1}T$ is at most $| \l_c^{-1}\l_2|<1$, where $|\l_2| < |\l_c|$ is the second largest eigenvalue of the $KS$ operator in modulus. Then, the function $f_n(z)=z^n$ converges uniformly to $0$ on $|z|  \le |\l_c^{-1}\l_2|$ and therefore, by Lemma 13 in \cite{DS} Chapter VII, we have $\parallel (\l_c^{-1}\K)^n - P \parallel_\x = \parallel (\l_c^{-1}T)^n \parallel_\x  = \parallel  f_n(\l_c^{-1}T) \parallel_\x \to 0$ as $n \to \infty$.
\end{proof}

\vspace{2mm}

%%%%%%%%%%%%%%%%%%%%%%%%%%%%%%%

	\begin{proof} \textit{Theorem \ref{teorema principal 2}}.

	   We want to show $|\l_c|=r(\K) \le \x^{-1}$. So, suppose the strict inequality $|\l_c| > \x^{-1}$ in order to obtain a contradiction.  Thus, let be  $\l_*$ any positive number such that $|\l_c| > \l_* > \x^{-1}$.

	    Now we claim the following.

	     \begin{quote}
	    	\textit{Claim: We have, }
	    	
	    	\begin{equation}\label{claim}
	    	\lim_{n \to \infty} \parallel \big(\l_*^{-1} \K\big)^n - P\parallel_\x < \infty.
	    	\end{equation} 
	    \end{quote}
	    
	   \textit{Proof of the Claim}:
	   
	   By definition of the projection $P$ in Theorem \ref{teorema principal}, the interior of the curve $\mathcal{C}$, except at $\l_c$, and the own $\mathcal{C}$ are in the resolvent set of the KS operator. So, of coarse,  $P$ is  bounded and there  exists a constant $J>0$ such that $\parallel P \parallel_\x <J$. Thus

	   \begin{equation}\label{sup 0}
	   \parallel  \big(\l_*^{-1} \K\big)^n  -P \parallel_\x  \le  \parallel \big(\l_*^{-1}\K\big)^n \parallel_\x + J .
	   \end{equation}
	   Let be $\vf = \big(\varphi\{z^m\}_{m \ge 1} (x)_l  \big)_{l \ge 1} \in \mathcal{D}_\x(\L) \subset \mathcal{E}_\x(\L)$.   Then, since  $U \ge -2B$  and $|z| < \x$, and because the action of $\K$,  see Remark (\ref{obs 2}),   we have,

	   \begin{equation}
	   |  \K^n \varphi\{z^m\}_{m \ge 1} (x)_l | \le \sum_{m \ge 0} {|z|^{m+l-n} \over m!}  |\int_{\L^m} d(y)_m  e^{-\b U\big((x)_l, (y)_m \big)}| < \x^{l-n} e^{2B} e^{\x\L}.
	   \end{equation}
	   Thus, for all $\vf \in \mathcal{D}_\x(\L)$,
	   
	   \begin{equation}\label{sup}
	   \parallel \big(\l_*^{-1} \K\big)^n \vf \parallel_\x = \sup_{l \ge 1} (\x^{-1})^{l}\l_*^{-n} ess \sup_{(x)_l}|  \K^n \varphi\{z^m\}_{m \ge 1} (x)_l | \le \sup_{l \ge 1} (\x^{-1})^{l}\l_*^{-n} \x^{l-n} e^{\x \L} = (\l_*^{-1} \x^{-1})^n  e^{\x \L}e^{2B}.
	   \end{equation}

	   Therefore (\ref{sup 0}), (\ref{sup})  and  $|\l_c|>\l_* > \x^{-1}$ imply  
	   \begin{equation}
	   \displaystyle \lim_{n \to \infty} \parallel \big(\l_*^{-1} \K\big)^n -P\parallel_\x \le J.
	   \end{equation}
	   The Claim is proved.
	 
%%%	 

       We see by  Lemma (\ref{lema2}) and the Claim (\ref{claim})  that for all $\e >0$   there is a  $n_0 \in \mathbb{N}$ such that for all $n>n_0$,

	    \begin{equation}\label{eq6}
	  \parallel  \big(\l_*^{-1} \K\big)^n - \big(\l_c^{-1} \K\big)^n  \parallel_\x \le  \parallel \big(\l_*^{-1} \K\big)^n - P \parallel_\x + \parallel P - \big(\l_c^{-1} \K\big)^n \parallel_\x =J +2\e.
	     \end{equation}
	      So, since $r(\K)^n=\l_c^n \le \parallel \K^n \parallel$, we have from (\ref{eq6}) above that  for  all $n>n_0$,
	     \begin{equation}\label{eq7}
	     |\l_*^{-n} - \l_c^{-n}| \le (J+2\e) \parallel \K^n \parallel_\x^{-1} \le (J+2\e) \l_c^{-n},	      \end{equation}
	      Finally,  this implies that the sequence $\{|(\l_c\l_*^{-1})^n -1|\}$ is bounded, an absurd if $|\l_c| > \l_* > \x^{-1}$. Therefore we have  $|\l_c| \le \x^{-1}$.

 	\end{proof}

\vspace{2mm}
%%%%%%%%%%%%%%%%%%%%%%%%%%%%%%%%%%%%%%%%%%%%%%%%%%%%%%%%%%%%%%%%%%%%%

\begin{proof}\textit{Theorem \ref{cor principal 2}}.
	 
	 Assume  $\x^{-1}=C$. From the Theorem \ref{teorema principal 2}  we have that  the spectral radius  of $\K$,  for any $\L$, is at most $C$, that is, $|z_c|^{-1}=|\l_c|=r(\K) \le C$. Then for all $n \ge1$  the thermodynamic limits, $\r_n(z)$,   are analytic functions of $z$, if $|z| < C^{-1}$. Indeed,  for any measurable compact set $\L \in \mathbb{R}^n$, if  $|\l|=|z|^{-1}>C \ge r(\K)$, then
	 
	 \begin{equation}\label{eq final}
	 \parallel \chi_\L [\r_{\L_l}(z,(x)_n)- \r_{\L_j}(z,(x)_n)] \parallel_\x \le  \sum_{n \ge 0} \parallel K^n_{\L_l}\a_{\L_l}-K^n_{\L_j}\a_{\L_j} \parallel_\x |\l|^{-n-1} < \infty.
	 \end{equation}
	 On the right hand side above, for all $n \ge 1$, the  term  $\parallel K^n_{\L_l}\a_{\L_l}-K^n_{\L_j}\a_{\L_j} \parallel_\x$ tends to zero as $l,j \to \infty$. So  the left side of (\ref{eq final})  tends to zero as $l,j \to \infty$,   on $|z| < C^{-1}$. Then    the sequence of correlation functions tends  to the thermodynamical limit $\r_n(z)$ uniformly on compact subsets of $|z| < C^{-1}$. Then, for each $n \in \mathbb{N} $,  $\r_n(z)$ is analytic on $|z| < C^{-1}$ and  therefore the radius of convergence of  $\r_n(z)$ is at least $C^{-1}$.  Note as well that  if  there is some zero such that its modulus is $C^{-1}=\x$, then the inequality for $r(\K)$ in the Theorem \ref{teorema principal 2} becomes equality, and we have the convergence radius  equal to $C^{-1}$. 
	 
\end{proof}

\begin{proof}\textit{Corollary \ref{cor principal 3}}.
	  
	  Let be $\mathcal{R}$ be the convergence radius of $\r_1(z)$.  For these particular cases of potentials we have, from  \cite{Ru} Chapter 4, Section 5 Theorem 4.5.3 (see inequalities 5.17, 5.18 and 5.19), that holds:
	  $$\mathcal{R} \le C^{-1},$$ and  for all $s>0$, $$\r_1 \ge {s\over  1+Cs}.$$ We also have that $\r_1(s)$ is an increasing function for $s>0$.

	Then, the proof of the Corollary (\ref{cor principal 3}) follows from  the Theorem \ref{cor principal 2} and the inequalities above. Indeed, the first one together theorems result, give us that $\mathcal{R}=C^{-1}$ and  the fact that the singularity of the cluster expansion is at $-\mathcal{R}$ comes from the \textit{alternating sign property} of its coefficients. Note that this implies, in fact, that $-C$ belongs to the spectrum. 	Another way to prove that $\mathcal{R}=C^{-1}$   is the following: the vector $\vf_\L(z)=\big( \varphi_\L \{z^m\}(x)_n\big)_{n \ge 1}$, 
	
	\begin{equation}
	\varphi_{\Lambda}\{z^m\}(x)_n=\sum_{m=0}^{\infty} \frac{z^{n+m}}{m!} \int_{\L^m} d(y)_m  e^{-\b U((x)_n,(y)_m)}, 
	\end{equation}
	for $\Phi>0$,  belongs to $\DD_\x$, in both cases  $\DD_\x(\L)=\mathcal{E}^s_\x(\L)$ and $\DD_\x(\L)=D_\x(\L)$ (see  Proposition \ref{prop 2 apend} in the Appendix), if and only if $|z| \le \x$. Then the same holds for $\Ro(z)= \X_\L^{-1}(z)\vf_\L(z)$. Then, because  Proposition \ref{prop 2}, we may not have a  bound on the spectral radius less than $\x^{-1}$. Therefore  we must have $|\l_c|=\x^{-1}=C$.
	
	To see the bound on the convergence radius of the  Virial expansion, note that if   $0<s <C^{-1}$ then $\r_1 (s) \ge s( 1+Cs)^{-1} \ge s2^{-1}$ and so we have that radius of convergence of $z$ in powers of $\r_1$ is at least $(2C)^{-1}$. This implies that the same holds  for the Virial expansion. Therefore    there does not exist phase transition if  $0<\r_1<(2C)^{-1}$.

\end{proof}

%%%%%%%%%%%%%%%%%%%%%%%%%%%%%%%%%%%%%%%%%%%%%%%%%%%%%%%%%%%%%%%

\section{Appendix}\label{apendice}

\begin{prop}\label{prop 1 apen}
	The subspace 
	 \begin{equation*}
	D_\x(\L)=\{\vf  =\big(\varphi_{\Lambda}\{a_m\}(x)_n\big)_{n \ge 1} \in \mathcal{E}_\x(\L):\,\,  \varphi_{\Lambda}\{a_m\}(x)_n =\sum_{m=0}^{\infty} \frac{a_{n+m}}{m!} \int_{\L^m} d(y)_m  e^{-\b U((x)_n,(y)_m)}< \infty 
	\end{equation*}
	\begin{equation}\label{subespace}
	\forall n \ge 1 , \forall (x)_n \in \mathbb{R}^{n\n}, \mbox{ and a complex sequence } \{a_m\}_{m=1}^{\infty} \}.
	\end{equation}
is a Banach space.
\end{prop} 

\begin{proof}

	Let  $\{\vf^q\}_{q \ge 1}=\{\vq\}_{q \ge 1}$ be a Cauchy sequence in $D_\x(\L)$, in which,
	
	\begin{equation}
	\vq = (\vq(x)_1, \vq(x)_2, \dots), 
	\end{equation}
	and
	\begin{equation}
	\vq (x)_n = \sum_{m=0}^{\infty} \frac{a^q_{n+m}}{m!} \int_{\L^m} d(y)_m  e^{-\b U((x)_n,(y)_m)} < \infty, \,\,\,\, \forall q,n \ge 1, \forall (x)_n \in \mathbb{R}^{n\n}.
	\end{equation}
Let be $\e>0$ and $n \in \mathbb{N}$. Because 	$\{\vf^q\}_{q \ge 1}$ is a Cauchy sequence and $\e/\x^n>0$, there is a  integer $q_0$ such that  if $q,l \ge q_0$, then 
$\parallel \vf^q - \vf^l \parallel_\x  < \e/\x^n$. So, in particular, for all $n \ge 1 $ and   almost every $(x)_n$,   if $q,l \ge q_0$, then 
	\begin{equation}\label{eq apendice}
	|\vq(x)_n - \vl(x)_n| < \x^n \parallel \vf^q - \vf^l \parallel_\x  < \e,   
	\end{equation} 
	that is, the sequence $\{\vq(x)_n\}_{q \ge 1}$ is a Cauchy sequence in $\mathbb{C}$,  for all $n \ge 1 $ and almost every $(x)_n$. Define  for all $n \ge 1 $ and almost every $(x)_n$,
	\begin{equation}
	f(x)_n=\lim_{q \to \infty} \vq(x)_n.
	\end{equation}
	Now, taking the limit $q \to \infty$ in (\ref{eq apendice}) we have that $\bm{f}=(f(x)_1, f(x)_2, \dots)$ is such that $\parallel \bm{f} - \vf^l \parallel_\x < \e$, if $l \ge q_0$, and (\ref{eq apendice}) also implies $\bm{f} \in \mathcal{D}_\x(\L)$.
\end{proof}

     Next we give a characterization of the subspace $D_\x$ of the former Proposition \ref{prop 1 apen}. It is similar to Lemma 2 in \cite{Pas} for the subspace considered there, which the next proposition is already implicit.

     \begin{prop}\label{prop 2 apend}
     	The vector $\vf_\L  =\big(\varphi_{\Lambda}\{a_m\}(x)_n\big)_{n \ge 1}$, with 
     	$$\varphi_{\Lambda}\{a_m\}(x)_n=\sum_{m=0}^{\infty} \frac{a_{n+m}}{m!} \int_{\L^m} d(y)_m  e^{-\b U((x)_n,(y)_m)},$$

     	belongs to $D_\x$, defined in (\ref{subespace}), if and only if the sequence $\{a_m\}_{m=1}^{\infty}$ satisfies, 
     	
     	\begin{equation}\label{sequence cond}
     	\sup_{m \ge 1} \frac{|a_m|}{\x^m} \mbox{ ess}\sup_{(x)_m \in \mathbb{R}^\n} |\exp^{-\b U((x)_m)}| < \infty.
     	\end{equation}
     \end{prop}
 
 \begin{proof}
 	Consider the vector $\vf'_\L=\big(a_n e^{-\b U((x)_n)} \big)_{n \ge 1}$. Observe that $\vf'_\L \in \mathcal{E}_\x(\L)$ if and only if 
 	\begin{equation}
 \parallel \vf'_\L \parallel_\x =	\sup_{m \ge 1} \frac{|a_m|}{\x^m} \mbox{ ess}\sup_{(x)_m \in \mathbb{R}^\n} |\exp^{-\b U((x)_m)}| < \infty.
 	\end{equation}
 	From the inequality $|a_n e^{-\b U((x)_n)}| \le \x^n \parallel \vf'_\L \parallel_\x < \infty$,  we have
 		$$a_n e^{-\b U((x)_n)}=\vf'_\L\{b_m\}(x)_n=\sum_{m=0}^{\infty} \frac{b_{n+m}}{m!} \int_{\L^m} d(y)_m  e^{-\b U((x)_n,(y)_m)} < \infty,$$
 	with $b_{n+m}=a_n$ or $0$, if $m=0$ or $m \ge 1$, respectively. So, if $\parallel \vf'_\L \parallel_\x < \infty$, then $\vf'_\L \in D_\x$.
 	
 	Now,  consider the following bounded operators in $D_\x$,  $A_\L$ and $B_\L$ defined in \cite{Pas}, 
 	
 		\begin{equation}
 	A_\L \vf =  \big( A_\L \varphi (x)_n\big)_{n \ge 1}, \,\,\,\,   	B_\L \vf =  \big( B_\L \varphi (x)_n\big)_{n \ge 1},
 	\end{equation}
 	where,
 	\begin{equation}
 	A_\L \varphi (x)_n = \sum_{m=0}^{\infty} \frac{1}{m!} \int_{\L^m} \varphi((x)_n, (y)_m) d(y)_m, \,\,\,\, B_\L \varphi (x)_n = \sum_{m=0}^{\infty} \frac{{(-1)}^m}{m!} \int_{\L^m} \varphi((x)_n, (y)_m) d(y)_m.
 	\end{equation}
 	We showed in the preceding paragraph that if $\parallel \vf'_\L \parallel_\x < \infty$, then $\vf'_\L \in D_\x$.  Because  it holds $A_\L \vf'_\L=\vf_\L$, we have that if  $\parallel \vf'_\L \parallel_\x < \infty$ then  $\vf_\L \in D_\x$, by the Closed Graph Theorem (CGT).   Reciprocally, by the CGT, if $\vf_\L \in D_\x$, then $\vf'_\L=B_\L \vf_\L$  belongs as well and so $\parallel \vf'_\L \parallel_\x < \infty$.
 \end{proof}

 \begin{prop}\label{prop 3 apend}
 	The Kirkwood-Salsburg operator $\K:\mathcal{D}_\x(\L) \to \mathcal{D}_\x(\L)$ is bounded, if we consider it on the setting (a) or (b) in Remark \ref{obs 1}.
 \end{prop}
 
 \begin{proof}
 	If we consider the KS operator defined under the conditions in (a), that is, $\mathcal{D}_\x(\L)=\mathcal{E}_\x^s$ with positive/hard core potentials, then we can see straightforward from definitions that $\K$ has finite norm  $\parallel \cdot \parallel_\x$ in $\mathcal{E}_\x \supset \mathcal{E}_\x^s$ and that $\K \mathcal{E}_\x^s \subset \mathcal{E}_\x^s$.
 	
 	Consider now the KS operator defined under the conditions in (b), that is, $\mathcal{D}_\x(\L)= D_\x(\L)$ with regular and stable potentials. For  a vector $\vf_\L \in D_\x$ defined by
 	
 	 $$\varphi_{\Lambda}\{a_m\}(x)_n=\sum_{m=0}^{\infty} \frac{a_{n+m}}{m!} \int_{\L^m} d(y)_m  e^{-\b U((x)_n,(y)_m)},$$
 	 we have by definition of the KS operator that the vector $\K\vf$ is given by
 	 
 	 $$\K \varphi_{\Lambda}\{a_m\}(x)_{n+1}=\sum_{m=0}^{\infty} \frac{a_{n+m}}{m!} \int_{\L^m} d(y)_m  e^{-\b U((x)_{n+1},(y)_m)},$$
 	 in which for $n=0$, the constant $a_0$ is given in the Remark \ref{obs 2}. Therefore, if $\vf_\L \in D_\x(\L)$ then $\K \vf_\L \in D_\x(\L)$, and this also proves that $\K$ has finite norm  $\parallel \cdot \parallel_\x$ by definition of $D_\x(\L)$.
 \end{proof}

%%%%%%%%%%%%%%%%%%%%%%%%%%%%%%%%%%%%%%%%%%%%%%%%%%%%%%%%%%%%

\section{Discussion}\label{dicussao}

          The study of the spectral properties of the Kirkwood-Salsburg operator have shown to be challenge due the  influence in the spectrum of the interaction potential as well as the operator domain. Fortunately, for representative situations like positive, or more general, stable and regular potentials, and a suitable domain operator, we have the good property to have a point spectrum equal to the set of inverse zeros of the partition function. This allow us to use different techniques to obtain precise bounds on the convergence radius, instead, for instance, to deal with the hard problem of find the exact form of the series coefficients.
          
           The spectrum of the Kirkwood-Salsburg operator formed by eigenvalues equal to inverse zeros of the partition function, give us a lot of similarities with compact operators. This fact was already discussed in Zagrebnov \cite{Z}, which the author shows, for instance, that for hard core potentials  $\K^n$ is compact for some $n$ big enough. Our Theorem \ref{teorema principal} in which we prove the Laurent expansion of the resolvent has indeed some similarities with compact operators which is known to have a similar expansion, see Steinberg \cite{S}.  From Theorem \ref{teorema principal}, we obtain the  key Lemma \ref{lema2}. The result of the lemma in  (\ref{decomp K 2}), confirms and expands for more general potentials, the results of \cite{Z}, see Remark 3.5 therein, that the KS operator is quasipotentially compact in the sense that $\K^n$ can be arbitrarily close to a compact operator  for some $n$ big enough, here the finite dimensional rank spectral projection $P$.    

The asymptotic result  (\ref{asymp 1}) in Theorem \ref{teorema principal} is new for classical continuous models and it shows  how correlation functions behave close  the singularity. Other interesting property proved is  the simplicity of the closest zero of the partition function. The prove of this fact seems to be new as far as we know. A  result related with this appears in Csikvári \cite{C} where the author proves, among other things, the simplicity of the smallest root of the Independence polynomial in which is related with the partition function of the Hard Core lattice model.    

  The main consequence of this work is, indeed, the lower bound of at least $C^{-1}$  of the convergence radius of the Cluster expansion of the density $\r_1$ for  stable and regular potentials, Theorem \ref{cor principal 2}. The equality for the convergence radius holds if there is a zero of the partition function with modulus $C^{-1}$.   In case of positive or hard core potentials, Corollary \ref{cor principal 3}, we have the convergence radius equal to $C^{-1}$.  From the  Corollary \ref{cor principal 3} we also have that the convergence radius of the Virial expansion for  positive or hard core potentials  is at least $0,5/C$, which improves  the  bounds $1/(1+e)C=0.27/C$ in Ruelle \cite{Ru} and  $1/eC=0,37/C$ in Gorzela\'{n}czyk \cite{Go}.

%%%%%%%%%%%%%%%%%%%%%%%%%%%%%%%%%%%%%%%%%%%%%%%%%%%%%%%%%%%%%

\vspace{8mm}

\textit{DEPARTAMENTO DE MATEMÁTICA - UFOP - OURO PRETO - MG - BRASIL}. 
 
\textit{Email address : rgalves@iceb.ufop.br}

\end{document}